\documentclass[12pt]{article}

\usepackage{amssymb,amsfonts,amsthm}
\usepackage{amsmath}
\usepackage{bm}             
\usepackage[mathscr]{eucal}
\usepackage{cancel}
\usepackage{wasysym}
\usepackage{array}

\usepackage{graphicx}
\usepackage{subfigure}
\def\be{\begin{equation}}
\def\ee{\end{equation}}
\def\bea{\begin{eqnarray}}
\def\eea{\end{eqnarray}}

\def\hsp5{\hspace{5mm}}

\theoremstyle{plain}
\newtheorem{theorem}{Theorem}[section]

\theoremstyle{remark}

\newtheorem{remark}{Remark}[section]

\newcommand{\sfrac}[2]{{\textstyle{#1\over#2}}}

\title{\sc Cosmological global dynamical systems analysis}

\begin{document}

\author{
\sc Artur Alho,$^{1}$\thanks{Electronic address:{\tt
aalho@math.ist.utl.pt}}\,, Woei Chet Lim,$^{2}$\thanks{Electronic address:{\tt
wclim@waikato.ac.nz}}\, and Claes Uggla,$^{3}$\thanks{Electronic
address:{\tt claes.uggla@kau.se}}\\
$^{1}${\small\em Center for Mathematical Analysis, Geometry and Dynamical Systems,}\\
{\small\em Instituto Superior T\'ecnico, Universidade de Lisboa,}\\
{\small\em Av. Rovisco Pais, 1049-001 Lisboa, Portugal.}\\
$^{2}${\small\em Department of Mathematics, University of Waikato,}\\
{\small\em Private Bag 3105, Hamilton 3240, New Zealand.}\\
$^{3}${\small\em Department of Physics, Karlstad University,}\\
{\small\em S-65188 Karlstad, Sweden.}}


\date{\normalsize{January 19, 2022}}
\maketitle

\begin{abstract}

We consider a dynamical systems formulation for models with an exponential
scalar field and matter with a linear equation of state in a spatially flat and isotropic
spacetime. In contrast to earlier work, which only considered linear hyperbolic fixed
point analysis, we do a center manifold analysis of the non-hyperbolic fixed points
associated with bifurcations. More importantly though, we construct monotonic functions
and a Dulac function. Together with the complete local fixed point analysis this leads to
proofs that describe the entire global dynamics of these models, thereby complementing 
previous local results in the literature.

\end{abstract}

\newpage

\section{Introduction\label{sec:intro}}

Dynamical systems and dynamical systems methods were introduced in
cosmology in 1971 by Collins~\cite{collins71} who treated 2-dimensional
dynamical systems while Bogoyavlensky and Novikov (1973) \cite{BN73}
used dynamical systems techniques for higher dimensional dynamical
systems in cosmology. This early work has subsequently
been followed up and extended by many researchers,
see e.g.~\cite{waiell97,col03,bahetal18}. The
first dynamical systems analysis involving a minimally coupled scalar field
was given by Belinski\v{\i} and
coworkers~\cite{beletal85a,beletal85b,belkha87,beletal88} who used dynamical systems
to explore inflation, primarily focusing on the potentials 
$V(\varphi) = \sfrac12 m^2\varphi^2$ and 
$V(\varphi)=\sfrac14 \lambda (\varphi^2-\varphi^2_0)^2$.
In 1987 Halliwell~\cite{hal87} treated an exponential scalar field potential
in Friedmann-Lema\^{i}tre-Robertson-Walker (FLRW) cosmology,
where the scalar field and an exponential representation of the
cosmological scale factor were used as dynamical systems variables,
which resulted in an unbounded state space and thereby only a local state
space description (see also Ratra and Peebles (1988)~\cite{ratpee88}
for early work using dynamical systems for this case).
The first global state space treatment of a scalar field with an exponential
potential was done in Bianchi cosmology, which contains FLRW cosmology
as a special case, by Coley \emph{et al.} (1997)~\cite{coletal97}.
Notable is also the work by Foster (1998)~\cite{fos98} who analysed asymptotically
exponential scalar field potentials. Finally,
Copeland \emph{et al.} (1998)~\cite{copetal98} treated the spatially flat
FLRW case with an exponential potential and a perfect fluid with a linear
equation of state. This latter work, which used a reduced 2-dimensional
compact and regular state space, gave a
linear fixed point (critical point, equilibrium point)
stability analysis for parameter values that yielded
hyperbolic fixed points.\footnote{All eigenvalues of a linearization
of a dynamical system at a so-called hyperbolic fixed point
have non-zero real parts, and hence such a fixed point has no center manifolds.}

In this paper we will investigate the models Copeland \emph{et al.} considered,
but we will extend their local fixed point analysis to the bifurcation
values of the relevant model parameters, which yield non-hyperbolic fixed
points with one zero eigenvalue, thereby requiring center manifold analysis.
More importantly though, even in the case of the hyperbolic fixed points a
linear fixed point analysis in a 2-dimensional compact state space does not
necessarily imply a complete asymptotic description; other asymptotic behaviour
is possible, such as periodic orbits and heteroclinic cycles.\footnote{A
heteroclinic orbit is a solution trajectory that originates and ends at
two different fixed points; a heteroclinic chain consists of a concatenation
of heteroclinic orbits, where the ending fixed point of one heteroclinic
orbit is the starting fixed point of the next one; a heteroclinic cycle
is a closed heteroclinic chain. For an example and detailed discussion of
a heteroclinic cycle arising from a scalar field potential in the spatially
flat FLRW case, see Foster (1998)~\cite{fos98b}.}
In this paper we fill these gaps in the proof for the asymptotic
and global behaviour of models with an exponential scalar field potential and a
perfect fluid with a linear equation of state.

We finally note two additional motivational points: First, a
substantial fraction of scalar field potentials used to describe inflation
or/and quintessence are asymptotically exponential when the scalar field
$\varphi\rightarrow+\infty$ or/and $\varphi\rightarrow-\infty$. A global
description of the solution space of such models therefore requires a global understanding of
the present models. Second, a key ingredient for several of the proofs in the present
paper are monotonic functions, which are \emph{derived} by using methods first
developed by Uggla in ch. 10 in~\cite{waiell97} and later generalized by
Uggla and coworkers in~\cite{heiugg10} and~\cite{heretal10}. The present
results thereby serve as an illustration of the power of those methods.

\section{Dynamical systems description and local fixed point analysis\label{sec:derivationdyn}}

\subsection{Field equations}

We consider a flat and isotropic FLRW
spacetime,
\begin{equation}
ds^2 = -dt^2 + a^2(t)\delta_{ij}dx^idx^j,
\end{equation}
where $a(t)$ is the cosmological scale factor. The source consists of
matter with an energy density $\rho_\mathrm{pf}>0$ and pressure
$p_\mathrm{pf}$, and a minimally coupled scalar field, $\varphi$, with a
potential $V(\varphi)>0$, which results in
\begin{equation}\label{rhophipphi}
\rho_\varphi = \frac12\dot{\varphi}^2 + V(\varphi),\qquad
p_\varphi = \frac12\dot{\varphi}^2 - V(\varphi).
\end{equation}

The Einstein equations, the (non-linear) Klein-Gordon
equation, and the energy conservation law for the fluid,
are given by\footnote{We use units such that $c=1$ and $8\pi G =1$,
where $c$ is the speed of light and $G$ is Newton's gravitational
constant.} (see, e.g., \cite{alhugg15b} for a slightly
different formulation of the equations)
\begin{subequations}\label{Mainsysdim}
\begin{align}
\dot{a} &= aH,\label{adotH}\\
\dot{H} + H^2 &= -\frac16(\rho + 3p), \label{Ray}\\
3H^2 &= \rho, \label{Gauss}\\
\ddot{\varphi} &=-3H\dot{\varphi} - V_{,\varphi}, \label{KG}\\
\dot{\rho}_\mathrm{pf} &= -3H\gamma\rho_\mathrm{pf},
\end{align}
\end{subequations}
where an overdot represents the derivative with respect to the cosmic
time $t$; a barotropic equation of state for the perfect fluid yields
$\gamma=\gamma(\rho_\mathrm{pf})$; the total energy density $\rho$ and
pressure $p$ are given by
\begin{equation}
\rho = \rho_\varphi + \rho_\mathrm{pf},\qquad
p = p_\varphi + p_\mathrm{pf}.
\end{equation}
Equation~\eqref{Ray} is the (Landau–-) Raychaudhuri equation,
while~\eqref{Gauss} is the Gauss/Hamiltonian constraint (often referred to
as the Friedmann equation in FLRW cosmology).
Here we are going to consider an exponential potential and a perfect fluid
with a linear equation of state, i.e.,
\begin{equation}\label{exppotlinjeqstate}
V = V_0e^{-\lambda\varphi},\qquad p_\mathrm{pf} = (\gamma-1)\rho_\mathrm{pf},
\end{equation}
where $\lambda$ and $\gamma$ are constants; matter, radiation and a stiff perfect fluid
correspond to $\gamma=1$, $\gamma=4/3$ and $\gamma=2$, respectively.
%
%

\subsection{Explicitly solvable cases}

For an exponential potential and a perfect fluid with linear equation of state
the equations are solvable for
several values of $\lambda$ and $\gamma$, as was shown implicitly
by Uggla {\it et al.} (1995)~\cite{uggetal95}
(use equations (2.23), (2.37), (4.98) and $\varphi=\sqrt{6}\beta^\dagger$
in Table III in~\cite{uggetal95}), where it was also demonstrated how
to obtain explicit solutions in as simple form as possible.\footnote{To the authors' knowledge,
all known explicit solutions for problems with hypersurface homogeneity, in general relativity
and modified gravity theories, are obtainable, and in their simplest form, by using the mechanisms
and methods in~\cite{uggetal95}.} The solvable cases are:
\begin{subequations}
\begin{alignat}{4}
\lambda &=0, &\qquad \gamma&=2, &\qquad & \\
\lambda &= \pm \sqrt{6}(\gamma-1), &\qquad \lambda &=\pm\left(\frac{10+\gamma}{2\sqrt{6}}\right),&\qquad \lambda &=\pm \sqrt{\frac32}\,\gamma,\\
\lambda &= \pm \sqrt{\frac32}(4-3\gamma), &\qquad \lambda &= \pm \left(\frac{4+\gamma}{\sqrt{6}}\right), &
\end{alignat}
\end{subequations}
where we recognize the first case, $\lambda=0$, as a constant potential
and the second case, $\gamma=2$, as a stiff perfect
fluid.\footnote{The explicitly solvable case $\gamma=1$, $\lambda=\sqrt{3/2}$
was used in~\cite{alhugg15b} to illustrate how explicit solutions
can be situated in a dynamical systems context.} A problem where either
the scalar field potential, $V$, or $\rho_\mathrm{pf}$ is zero yields a
trivially solvable problem, see~\cite{uggetal95}.

\subsection{The dynamical system}

To obtain a useful dynamical system, we introduce the following
\emph{dimensionless bounded} quantities\footnote{The variable
$\Sigma_\varphi$ was first introduced by Coley
\emph{et al.} (1997)~\cite{coletal97} and Copeland
\emph{et al.} (1998)~\cite{copetal98} whose $x$ is $\Sigma_\varphi$.
Since then, $\Sigma_\varphi$ (or $\varphi^\prime$)
is often used to describe scalar fields in cosmology,
see, e.g., Urena-Lopez (2012)~\cite{ure12}, equation (2.3),
Tsujikawa (2013)~\cite{tsu13}, equation (16) and Alho and Uggla
(2015)~\cite{alhugg15b}, equation (8). The reason for using the
notation $\Sigma$ for the kernel is because $\Sigma_\varphi$ plays
a similar role as Hubble-normalized shear, which is typically denoted
with the kernel $\Sigma$, see e.g.~\cite{waiell97}.}
\begin{subequations}\label{vardef1}
\begin{align}
\Sigma_\varphi &\equiv \frac{\dot{\varphi}}{\sqrt{6}H} = \frac{\varphi^\prime}{\sqrt{6}}, \label{Sigvarphidef}\\
\Omega_V &\equiv \frac{V}{3H^2},\label{OmVdef}\\
\Omega_\mathrm{pf} &\equiv \frac{\rho_\mathrm{pf}}{3H^2},\label{Ommdef}
\end{align}
\end{subequations}
A ${}^\prime$ henceforth denotes the derivative
with respect to $e$-fold time
\begin{equation}\label{Ndef}
N \equiv \ln{a/a_0},
\end{equation}
where $a_0 = a(t_0)$, $t=t_0 \Rightarrow N=0$. The definition~\eqref{Ndef}
implies that $N\rightarrow - \infty$ and $N\rightarrow + \infty$
when $a\rightarrow 0$ and $a \rightarrow\infty$, respectively.

Throughout we replace $t$ with $N$ by using that
\begin{equation}
\frac{d}{dt}=H\frac{d}{dN},\qquad \frac{d^2}{dt^2}=
H^2\left(\frac{d^2}{dN^2} - (1+q)\frac{d}{dN}\right),
\end{equation}
where
\begin{equation}\label{qdef1}
q \equiv -\frac{a\ddot{a}}{\dot{a}^2} = -1 - \frac{H^\prime}{H}
\end{equation}
is the \emph{deceleration parameter}.

Using $N$ and inserting~\eqref{exppotlinjeqstate} and the definitions~\eqref{vardef1}
into~\eqref{Mainsysdim} results in the following coupled system
for the state vector $(\Sigma_{\varphi},\Omega_\mathrm{pf})$:
\begin{subequations}\label{exppf}
\begin{align}
\Sigma_\varphi^\prime &= - (2-q)\Sigma_\varphi + \sqrt{\frac32}\,\lambda\Omega_V, \label{SigmaExpBound}\\
\Omega_\mathrm{pf}^\prime &= [2(1+q) - 3\gamma]\Omega_\mathrm{pf}, \label{OmegaExpBound}
\end{align}
\end{subequations}
where
\begin{subequations}
\begin{align}
\Omega_V &= 1 - \Sigma_\varphi^2 -  \Omega_\mathrm{pf},\label{constrGauss}\\
q &= - 1 + 3\Sigma_\varphi^2 + \frac32\,\gamma\Omega_\mathrm{pf}
= 2 - 3\Omega_V - \frac32\left(2 - \gamma\right)\Omega_\mathrm{pf}.\label{qpf1}
\end{align}
\end{subequations}
Restricting $\gamma$ to $\gamma\in(0,2)$, as we will do later, it follows that
$-1\leq q \leq 2$, where $q=-1$ when $\lambda=0$ and $\Sigma_\varphi=\Omega_\mathrm{pf}=0$,
while $q = 2$ when $\Sigma_\varphi=\pm 1$, $\Omega_\mathrm{pf}=0$. Another quantity that is often
used in the context of scalar field is $w_\varphi \equiv p_\varphi/\rho_\varphi$, or, equivalently,
$\gamma_\varphi$, defined by $p_\varphi = (\gamma_\varphi - 1)\rho_\varphi$ and hence
\begin{equation}
\gamma_\varphi \equiv \frac{p_\varphi + \rho_\varphi}{\rho_\varphi} = \frac{2\Sigma_\varphi^2}{1 - \Omega_\mathrm{pf}}.
\end{equation}

We use~\eqref{constrGauss} in~\eqref{exppf} to globally solve for $\Omega_V$,
although note that
\begin{equation}
\Omega_V^\prime = 2\left(1 + q - \sqrt{\frac32}\lambda\Sigma_\varphi\right)\Omega_V,
\end{equation}
which follows from~\eqref{constrGauss} and~\eqref{exppf}.
This equation and~\eqref{OmegaExpBound}
show that $\Omega_V=0$ and $\Omega_\mathrm{pf}=0$
form an invariant boundary of the state space $(\Sigma_\varphi,\Omega_\mathrm{pf})$,
which, due to that the dynamical system~\eqref{exppf} is completely regular, can be included in the
state space analysis. This is essential since some of the asymptotics
are associated with this boundary. We will refer to the orbits with $\Omega_V>0$, $\Omega_\mathrm{pf}>0$
as \emph{interior orbits} and orbits with $\Omega_V = 1 - \Sigma_\varphi^2 - \Omega_\mathrm{pf} = 0$
or/and $\Omega_\mathrm{pf}=0$,
as boundary orbits.

The present formulation
can be viewed as a transformation of an original state space $(H,\rho_\mathrm{pf}, \dot{\varphi}, \varphi)$
(alternatively, $(H,a, \dot{\varphi}, \varphi)$, since $\rho_\mathrm{pf} \propto a^{-3\gamma}$)
to $(H,\varphi, \Sigma_\varphi, \Omega_\mathrm{pf})$.
The equations for $H$ and $\varphi$, $H^\prime = - (1+q)H$ and $\varphi^\prime = \sqrt{6}\Sigma_\varphi$,
decouple from the dynamical
system for $(\Sigma_\varphi,\Omega_\mathrm{pf})$. This reduced state space can be therefore be
regarded as a projection of the state
space $(H,\varphi, \Sigma_\varphi, \Omega_\mathrm{pf})$.\footnote{More precisely, the new variables result in a
skew-product dynamical system where the base dynamics acts in $(\Sigma_\varphi,\Omega_\mathrm{pf})$
while the fiber dynamics acts in $(H,\varphi)$, a notion that was introduced 
in~\cite{anz51}.} The reason for the decoupling of $H$ and $\varphi$
is due to the linear equation of state for the perfect fluid and that
$- V_{,\varphi}/V = \lambda = \mathrm{constant}$.\footnote{The present system is closely connected to that of
Copeland \emph{et al.} (1998)~\cite{copetal98} who used $x=\Sigma_\varphi$ and $y=\sqrt{\Omega_V}$ as variables.
We prefer to use the more physical variable $\Omega_\mathrm{pf}$ rather than $y$. Moreover,
note that in contrast to $\Omega_\mathrm{pf}$, the unfortunately widely used $y$ is unsuitable for many
more general potentials. For examples where $y$ is inappropriate and for proper choices of variables, see,
e.g.,~\cite{alhugg15,alhetal15,alhugg17,alhetal21}.} Since the decoupled
equations can be solved by quadratures once $\Sigma_\varphi(N)$ and $\Omega_\mathrm{pf}(N)$ are obtained,
the system for the state vector $(\Sigma_\varphi,\Omega_\mathrm{pf})$ contains the essential
information for the present problem.

\subsection{Local hyperbolic fixed point analysis}

The fixed points of the dynamical system~\eqref{exppf} and the eigenvalues
of the linearization at the fixed points
are given in Table~\ref{tab:fixed points1}.
\begin{table}
\begin{center}
\begin{tabular}{|m{1cm}|m{1.5cm}|m{1cm}|m{1cm}|m{1.8cm}|m{4.5cm}|}\hline
Name & $\Sigma_\varphi$ & $\Omega_\mathrm{pf}$ & $\gamma_\varphi$ & $q$ & Eigenvalues \\ \hline
$\mathrm{K}_+$ & 1 & 0 & $2$ & $2$ &
$3(2-\gamma)$;\quad $\sqrt{6}(\sqrt{6}-\lambda)$\\ \hline
$\mathrm{K}_-$ & $- 1$ & 0 & $2$ & $2$ &
$3(2-\gamma)$;\quad $\sqrt{6}(\sqrt{6}+\lambda)$\\ \hline
$\mathrm{P}/\mathrm{dS}$ & $\frac{\lambda}{\sqrt{6}}$ & 0 & $\frac{\lambda^2}{3}$ & $\sfrac12(\lambda^2-2)$ &
$-\sfrac12(6-\lambda^2)$;\quad $-(3\gamma-\lambda^2)$\\ \hline
$\mathrm{FL}$ & 0 & 1 & $-$ & $\sfrac12(3\gamma-2)$ & $3\gamma$;\quad $-\sfrac32(2-\gamma)$\\ \hline
$\mathrm{S}$ & $\sqrt{\frac32}\left(\frac{\gamma}{\lambda}\right)$ & $1 - \frac{3\gamma}{\lambda^2}$ & $\gamma$ & $\sfrac12(3\gamma-2)$ &
$-\sfrac34(2-\gamma)(1 \pm \sqrt{r}\,)$\\ \hline
\end{tabular}
\caption{\label{tab:fixed points1} Fixed points and their effective scalar field equation of state $\gamma_\varphi$;
deceleration parameter $q$; their eigenvalues, where
$r \equiv 1 - \frac{8\gamma(\lambda^2-3\gamma)}{\lambda^2(2-\gamma)} = \frac{24\gamma^2 - (9\gamma-2)\lambda^2}{\lambda^2(2-\gamma)}$.}
\end{center}
\end{table}

The names of the fixed points are motivated as follows:
$\mathrm{K}_\pm$ are the boundary `kinaton' fixed points, due to that $\Omega_V=0$
and $\Omega_\mathrm{pf}=0$ and hence that
$H^\prime = -(1+q)H= -3H \Rightarrow
3H^2 = \rho =\rho_\varphi = 3H_0^2\exp(-6N) = 3H_0^2(a_0/a)^6$,
which characterizes kinaton evolution (a nomenclature
introduced in~\cite{joypro98}); $\mathrm{dS}$ with $\lambda=0$ and $q=-1$
is the de Sitter fixed point while $\mathrm{P}$, which exists when $0 < \lambda^2 < 6$
and yields power law acceleration when $\lambda^2 < 2$, due to that
$q=(\lambda^2 - 2)/2$ at $\mathrm{P}$; the fixed point $\mathrm{FL}$
is referred to as the Friedmann-Lema\^{i}tre fixed point ($\Omega_\mathrm{pf}=1$);
finally $\mathrm{S}$ is the scaling fixed point
(scaling due that $\gamma_\varphi = \gamma$ at $\mathrm{S}$, since this implies 
that the scalar field mimics the dynamics of the fluid, with a
constant ratio between both energy densities).

Apart from the de Sitter
fixed point $\mathrm{dS}$, which exists when $\lambda=0$, each of the other fixed points
correspond to a unique self-similar (i.e., the corresponding spacetime
admits a homothetic Killing vector field) power law solution, invariant under constant
conformal scalings. On the other hand, the de Sitter fixed point $\mathrm{dS}$
corresponds to a one-parameter set of solutions, parametrized by the dimensional
constant $V_0=\Lambda=3H_0^2$.

Without loss of generality, we will assume that $\lambda\geq 0$
(if $\lambda<0$, make the change $\varphi\rightarrow -\varphi$).
We also limit the range of $\gamma$ so that $\gamma\in(0,2)$
where $\gamma=0$ and $\gamma=2$ yield bifurcations, which is not surprising
since $\gamma=0$ results in a cosmological constant while a stiff fluid
equation of state, $\gamma=2$, corresponds to that the speed of sound is equal
to that of light (also, recall that $\gamma= 2$ is an explicitly solvable case).
The above eigenvalues then yield the stability properties
given in Table~\ref{tab:fixed points2}.
\begin{table}
\begin{center}
\begin{tabular}{|m{1cm}|m{2cm}|m{6cm}|}\hline
Name & Domain & Stability \\ \hline
$\mathrm{K}_+$ & $\lambda \geq 0$ &
Unstable node when $\lambda < \sqrt{6}$ \newline Saddle point for $\lambda>\sqrt{6}$\\ \hline
$\mathrm{K}_-$ & $\lambda \geq 0$ &
Unstable node
\\ \hline
$\mathrm{P}/\mathrm{dS}$ & $0 \leq \lambda^2<6$ &
Stable node when $\lambda^2 < 3\gamma$ \newline Saddle point for $3\gamma < \lambda^2 < 6$\\ \hline
$\mathrm{FL}$ & $\lambda\geq 0$ & Saddle point\\ \hline
$\mathrm{S}$ & $\lambda^2>3\gamma$ &
Stable node when $3\gamma<\lambda^2<\frac{24\gamma^2}{9\gamma-2}$
\newline Stable spiral for $\lambda^2>\frac{24\gamma^2}{9\gamma-2}$\\ \hline
\end{tabular}
\caption{\label{tab:fixed points2} Fixed points and their stability; $\gamma\in(0,2)$.}
\end{center}
\end{table}

It follows that there are three disjoint parameter regions when $\gamma\in(0,2)$, $\lambda \geq 0$, determined by
the two bifurcations at $\gamma = \lambda^2/3$ and $\lambda = \sqrt{6}$, see Figure~\ref{fig:Bifurcation}:
\begin{itemize}
\item[I:] $\gamma>\lambda^2/3$.
\item[II:] $\gamma<\lambda^2/3$, $\lambda<\sqrt{6}$.
\item[III:] $\lambda>\sqrt{6}$.
\end{itemize}

\begin{figure}[ht!]
	\begin{center}
		\includegraphics[width=0.4\textwidth, trim = 0cm 0cm 0cm 0cm ]{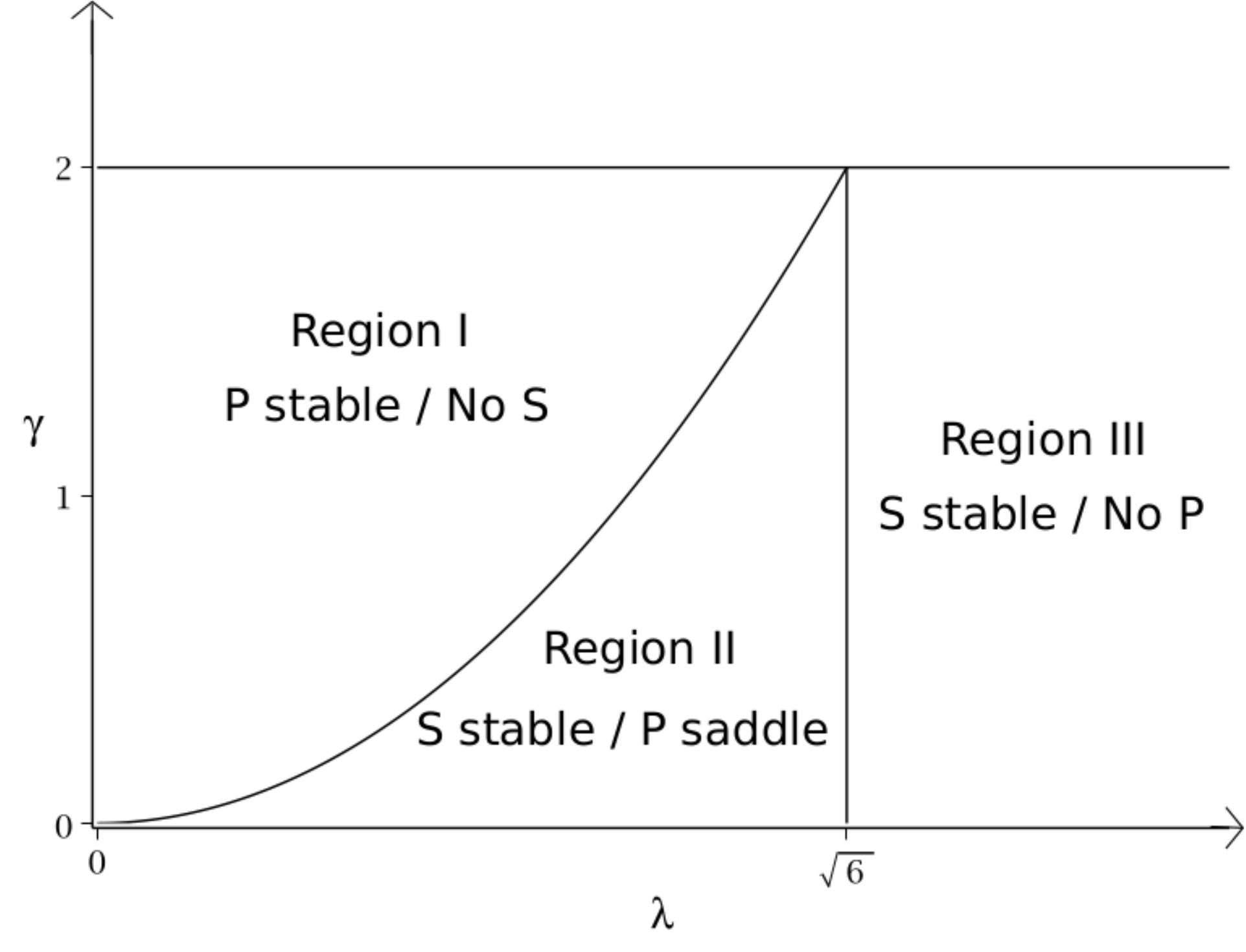}
		\vspace{-1.0cm}
	\end{center}
	\caption{Bifurcation diagram $(\gamma,\lambda)$.}
	\label{fig:Bifurcation}
\end{figure}

The $\lambda=0$ boundary of region I, where the stable
fixed point $\mathrm{P}$ is replaced with the stable fixed point
$\mathrm{dS}$ is, as mentioned, completely solvable. The solution
in the state space $(\Sigma_\varphi,\Omega_\mathrm{pf})$ can be obtained as follows.
In this case, due to that $\lambda=0$, the system~\eqref{exppf} is invariant
under $\Sigma_\varphi \rightarrow - \Sigma_\varphi$. Hence $\Sigma_\varphi$
can be replaced with $\Sigma_\varphi^2 = \Omega_\mathrm{stiff}$.
This results in a system that is identical to that for
a source with three matter components: (i) a
perfect fluid with $p_\mathrm{pf} = (\gamma_\mathrm{pf}-1)\rho_\mathrm{pf}$, (ii) a
stiff fluid, i.e., a perfect fluid with an equation of state
$p_\mathrm{stiff} = \rho_\mathrm{stiff}$ (and hence
$\rho_\mathrm{stiff}^\prime = -6\rho_\mathrm{stiff}$), and (iii) a cosmological constant,
$\Lambda = V$. This problem easily yields the solution
\begin{equation}\label{stiffmV}
(\Omega_\mathrm{stiff},\Omega_\mathrm{pf}) =
\frac{\left(\Omega_{\mathrm{stiff},0}e^{-6N},\Omega_{\mathrm{pf},0}e^{-3\gamma N}\right)}
{\Omega_{\mathrm{stiff},0}e^{-6N} + \Omega_{\mathrm{pf},0}e^{-3\gamma N} + \Omega_{V,0}},
\end{equation}
while $\Sigma_\varphi = \pm \sqrt{\Omega_\mathrm{stiff}}$ and
$\Omega_V = 1 - \Omega_\mathrm{stiff} - \Omega_\mathrm{pf}$.
The invariant subset $\Sigma_\varphi=0$,
and hence $\Omega_\mathrm{stiff}=0$, corresponds to having a perfect
fluid with a linear equation of state and a cosmological constant, where
$\gamma = 1$ yields the $\Lambda$CDM model.
For a visual representation of the orbit structure for the
$\lambda=0$ models, see Figure~\ref{fig:lambda0}.
\begin{figure}[ht!]
\begin{center}
	\includegraphics[width=0.4\textwidth, trim = 0cm 2cm 0cm 0cm ]{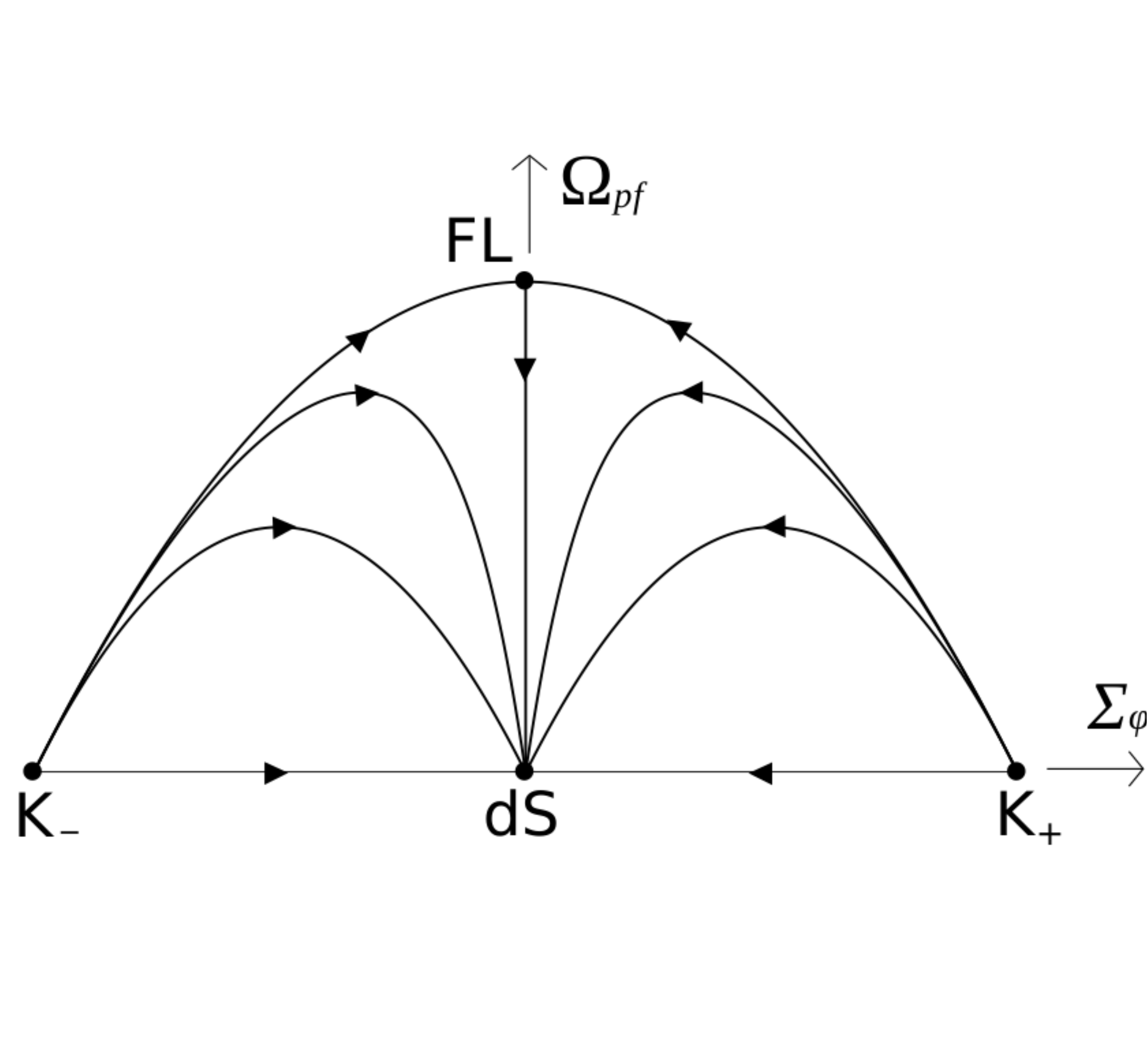}
	\vspace{-1.0cm}
\end{center}
\caption{Orbit structure for a perfect fluid and a constant scalar field potential and hence $\lambda=0$.
The heteroclinic orbit $\mathrm{FL}\rightarrow\mathrm{dS}$, which is the unstable manifold of
$\mathrm{FL}$ and also a separatrix
in the state space $(\Sigma_\varphi,\Omega_\mathrm{pf})$, yields the $\Lambda$CDM model
when $\gamma=1$.}
\label{fig:lambda0}
\end{figure}

The bifurcation boundary between region I and II at $\gamma=\lambda^2/3$ corresponds to
that $\mathrm{S}$ enters the state space through $\mathrm{P}$
and takes over as the stable sink when $\gamma<\lambda^2/3$
instead of $\mathrm{P}$, which is a stable sink in region I and a saddle with one orbit
entering the interior of the state space in region II.
The bifurcation boundary between regions II and III at $\lambda=\sqrt{6}$
corresponds to that $\mathrm{P}$ leaves the state space through $\mathrm{K}_+$,
where the latter is transformed from a source  (unstable node) to a saddle,
with no orbits entering the interior state space.
Finally, $\mathrm{K}_-$ is a source for all regions (and their bifurcation boundaries),
while $\mathrm{FL}$ is a saddle, with a single orbit entering the interior
state space, where $\mathrm{S}$ coalesce with $\mathrm{FL}$
in the limit $\lambda\rightarrow\infty$.

\subsection{Center manifold analysis of the non-hyperbolic fixed points}

Locally it remains to establish what happens near the fixed points
$\mathrm{P}/\mathrm{S}$ and $\mathrm{K}_+/\mathrm{P}$
at the bifurcation values $\gamma=\lambda^2/3$ and $\lambda=\sqrt{6}$, respectively.
In these cases one of the eigenvalues is zero. To establish what is happening
locally therefore requires a center manifold analysis.

The linearisation around $\mathrm{P}/\mathrm{S}$ yields the tangent spaces
\begin{subequations}\label{TSpacesPS}
\begin{align}
E^s &= \{(\Sigma_\varphi,\Omega_\mathrm{pf})|\, \Omega_\mathrm{pf} = 0\},\\
E^c &= \left\{(\Sigma_\varphi,\Omega_\mathrm{pf})|\, (\Sigma_\varphi-\sfrac{\lambda}{\sqrt{6}}) + \sfrac{\lambda}{\sqrt{6}}\,\Omega_\mathrm{pf} =0\right\},\label{Ec}
\end{align}
\end{subequations}
where $E^s$ is the stable tangent space spanned by the eigenvector associated
with the negative eigenvalue $-\sfrac12(6-\lambda^2)$, with the $\Omega_\mathrm{pf}=0$
axis being the invariant stable manifold $W^s$, while $E^c$ is the center tangent space
associated with the zero eigenvalue.
To study the stability associated with the zero eigenvalue we make use of the center
manifold reduction theorem (see e.g.~\cite{per00} section 2.12). Adapting the variables to $E^s$ and $E^c$
according to
%
\begin{equation}
(u,v)=\left((\Sigma_\varphi-\sfrac{\lambda}{\sqrt{6}})
+ \sfrac{\lambda}{\sqrt{6}}\,\Omega_\mathrm{pf},\Omega_\mathrm{pf}\right)
\end{equation}
result in that~\eqref{exppf} yields a dynamical system on the form
\begin{subequations}\label{backgroundeqcuv}
	\begin{align}
	u^\prime = -(3-\sfrac{\lambda^2}{2})u+	\mathcal{O}(\|(u,v)\|^2),\qquad
v^\prime = \mathcal{O}(\|(u,v)\|^2)
	\end{align}
\end{subequations}
%
where the fixed point $\mathrm{P}/\mathrm{S}$ is located at the origin $(u,v)=(0,0)$.
The analytical 1-dimensional center manifold $W^c$ can be represented locally as the
graph $h:\, E^c\rightarrow E^s$, where $u = h(v)$ satisfies the fixed point and tangency
conditions $h(0) =0$ and $\left.\frac{dh}{dv}\right|_{v=0} = 0$, respectively.
Inserting $u=h(v)$ into~\eqref{backgroundeqcuv} and using $v$ as the independent variable leads to
\begin{equation}\label{hv}
\left[2(1 + q) - \lambda^2\right]v\left(\frac{dh}{dv} -\frac{\lambda}{\sqrt{6}}\right) +
(2-q)g(v) - \sqrt{\frac32}\,\lambda \left(1 - v - g^2(v)\right) = 0,
\end{equation}
where
\begin{equation}
g(v)\equiv \frac{\lambda}{\sqrt{6}}(1-v) + h(v), \qquad
q=-1 + 3g^2(v) + \frac{\lambda^2}{2}\,v.
\end{equation}
This equation can be solved approximately by representing $h(v)$ as the formal power series
\begin{equation}
h(v) = \sum_{i=2}^n a_iv^i + {\cal O}(v^{n+1}) \qquad \text{as}\qquad v\rightarrow 0.
\end{equation}
Inserting
the above into~\eqref{hv} and solving algebraically for the coefficients leads to that on the center manifold
\begin{equation}
v^\prime = -\lambda^2 v^2+	\mathcal{O}(v^3) \qquad \text{as}\qquad v\rightarrow 0,
\end{equation}
which shows that the fixed point $\mathrm{P}/\mathrm{S}$ is stable. Thus a
1-parameter family of orbits converge to $\mathrm{P}/\mathrm{S}$ as $N\rightarrow+\infty$,
tangentially to the center subspace $E^c$. In terms of the original state-space
variables $(\Sigma_\varphi,\Omega_\mathrm{pf})$ the analytical center manifold expansion gives
\begin{equation}\label{SigmaCMho}
\begin{split}
\Sigma_\varphi &= \frac{\lambda}{\sqrt{6}}\left(1-\Omega_\mathrm{pf} + \frac{\lambda^2}{\lambda^2-6}\Omega^2_\mathrm{pf} + \mathcal{O}(\Omega^3_\mathrm{pf})\right)\qquad \text{as}\qquad \Omega_\mathrm{pf}\rightarrow 0.
\end{split}
\end{equation}

It remains to analyze the bifurcation fixed point $\mathrm{K}_+/\mathrm{P}$.
This turns out to be trivial since the center manifold of $\mathrm{K}_+/\mathrm{P}$
is the $\Omega_\mathrm{pf}=0$ axis, which is also a 1-dimensional unstable manifold of $\mathrm{K}_-$.
It follows that $\mathrm{K}_+/\mathrm{P}$ is a saddle, for which no interior orbits converge to
or originate from.

\section{Global dynamical systems analysis\label{app:exppot}}

In this section we first present monotonic functions that completely determine
the global solution structure of the present models. We then give an alternative
proof for the global dynamics by means of a Dulac function.

%
%


\subsection{Monotonic functions}

We here use the methods developed in ch. 10 in~\cite{waiell97} and then
generalized in~\cite{heiugg10} and~\cite{heretal10} to derive monotonic functions
that determine the global behaviour of the solution space of the present set of models.
The fixed point $\mathrm{P}/\mathrm{dS}$ is stable when $\gamma > \lambda^2/3$ (region I)
while the fixed point $\mathrm{S}$ is stable when $\gamma < \lambda^2/3$ (regions II and III and their mutual
boundary at $\lambda=\sqrt{6}$). For each of these two cases, and for $\gamma = \lambda^2/3$, there exists
a monotonic function associated with the stable fixed point.

We begin with the first case, i.e., region I where $\gamma > \lambda^2/3$ and $\mathrm{P}/\mathrm{dS}$ is
stable. We also include the bifurcation boundary $\lambda^2=3\gamma$ where $\mathrm{S}$ enter the state
space at $\mathrm{P}$. Below we will prove the following theorem:

\begin{theorem}\label{TheoremI}
{\bf Global interior dynamics when $\gamma \geq \lambda^2/3$.}
\begin{enumerate}
\item[(i)] All interior orbits end at $\mathrm{P}$ ($\mathrm{dS}$ when $\lambda=0$, $\mathrm{P}/\mathrm{S}$ when $\gamma = \lambda^2/3$).
\item[(ii)] A single interior orbit originates from $\mathrm{FL}$.
\item[(iii)] All remaining interior orbits originate from $\mathrm{K}_-$
and $\mathrm{K}_+$, each being a source for a 1-parameter set of interior orbits.
\end{enumerate}
\end{theorem}
\begin{remark}
Recall that $\Omega_V>0$ and $\Omega_\mathrm{pf}>0$ for interior orbits.
For a visual representation of the global orbit structure, see Figure~\ref{fig:globalI}.
\end{remark}
\begin{proof}
Using the methods in~\cite{heiugg10,heretal10} we derive
the following monotonic function for region I:
\begin{equation}
M \equiv (1-u\Sigma_\varphi)^2\Omega_V^{-1},\qquad u \equiv \frac{\lambda}{\sqrt{6}}<1,
\end{equation}
which is strictly monotonically decreasing for all interior orbits since
\begin{equation}
M^\prime = -\left(\frac{6(\Sigma_\varphi-u)^2 + (3\gamma-\lambda^2)\Omega_\mathrm{pf}}{1-u\Sigma_\varphi}\right)M <0,
\end{equation}
where $M$ takes its minimum value, $M=1-u^2$, at $\mathrm{P}/\mathrm{dS}$.
Thus all interior orbits end at $\mathrm{P}/\mathrm{dS}$.
Going backwards in time $M\rightarrow\infty$, which implies that $\Omega_V\rightarrow 0$. Taking the boundary structure
into account together with the local fixed point analysis shows that one orbit originates from $\mathrm{FL}$
while all other interior orbits originate from the fixed points
$\mathrm{K}_-$ and $\mathrm{K}_+$.

The bifurcation value $\gamma = \lambda^2/3$ yields
\begin{subequations}
\begin{align}
M^\prime &= -\left(\frac{6(\Sigma_\varphi-u)^2}{1-u\Sigma_\varphi}\right)M \leq 0,\\
M''|_{\Sigma_\varphi=u} &= 0,\\
M'''|_{\Sigma_\varphi=u} &= -108\Omega_\mathrm{pf}^2u^2(1-u^2)^3\Omega_V^{-1},\qquad \Omega_V = 1 - u^2 - \Omega_\mathrm{pf}.\label{Mppp}
\end{align}
\end{subequations}
As a consequence $M^\prime<0$ when $\Sigma_\varphi\neq u$. Furthermore, as follows from~\eqref{Mppp},
when $\Sigma_\varphi = u$ interior orbits only go through an inflection point since $\Omega_\mathrm{pf}>0$
(i.e., there is no invariant interior set with $\Omega_\mathrm{pf}>0$, $\Sigma_\varphi = u$). Thus $M\rightarrow 1-u^2$
toward the future with the limit at $\mathrm{P}/\mathrm{S}$, while $M\rightarrow\infty$ and hence $\Omega_V\rightarrow 0$ toward the past
also in this case. Combining this with the boundary structure and the local analysis of the hyperbolic fixed points $\mathrm{FL}$,
$\mathrm{K}_\pm$ and the center manifold analysis of $\mathrm{P}/\mathrm{S}$ yield the same result as
when $\gamma>\lambda^2/3$, which concludes the proof of Theorem~\ref{TheoremI}.
\end{proof}
\begin{theorem}\label{TheoremII}
{\bf Global interior dynamics when $\gamma < \lambda^2/3$.}
\begin{enumerate}
\item[(i)] All interior orbits end at $\mathrm{S}$.
\item[(ii)] A single interior orbit originates from $\mathrm{FL}$.
\item[(iii)] A 1-parameter set of interior orbits originate from
$\mathrm{K}_-$.
\item[(iv)] When $\lambda < \sqrt{6}$ there is also a 1-parameter set
of interior orbits that originates from $\mathrm{K}_+$, while a single
interior orbit originates from $\mathrm{P}$.
\item[(v)] When $\lambda \geq \sqrt{6}$ all interior orbits apart from the
single one from $\mathrm{FL}$ originate from $\mathrm{K}_-$.
\end{enumerate}
\end{theorem}
\begin{remark}
For a visual representation of the global orbit structure, see Figures~\ref{fig:globalII}
and~\ref{fig:globalIII}.
\end{remark}
\begin{proof}
Using the methods developed in ch. 10 in~\cite{waiell97} and~\cite{heiugg10,heretal10}
result in the following monotonic function when $\mathrm{S}$ is stable:
\begin{equation}
M \equiv (1 - v\Sigma_\varphi)^2\Omega_V^{a-1}\Omega_\mathrm{pf}^{-a}, \qquad v \equiv \sqrt{\frac32}\left(\frac{\gamma}{\lambda}\right),
\qquad a \equiv \frac{2(\lambda^2-3\gamma)}{2\lambda^2-3\gamma^2},
\end{equation}
where $0<v<1$ is the value of $\Sigma_\varphi$ at the fixed point $\mathrm{S}$,
and where we note that $0<a<1$ when $\gamma<\lambda^2/3$; thus, the two exponents of $\Omega_V$
and $\Omega_\mathrm{pf}$ are thereby negative. The $e$-fold time derivative
of $M$ is given by
\begin{equation}
M^\prime = -\frac{3(2-\gamma)(\Sigma_\varphi-v)^2 M}{(1-v^2)(1 - v\Sigma_\varphi)}\leq 0,
\end{equation}
where $M^\prime = 0$ when $\Sigma_\varphi=v$, but then $M''|_{\Sigma_\varphi=v}=0$ and
\begin{equation}
M'''|_{\Sigma_\varphi=v} = -\frac{27(2-\gamma)[2v^2 - \gamma(1-\Omega_\mathrm{pf})]^2M|_{\Sigma_\varphi=v}}{2v^2},
\end{equation}
where $M|_{\Sigma_\varphi=v} = (1-v^2)\Omega_V^{a-1}\Omega_\mathrm{pf}^{-a}$,
$\Omega_V = 1 - v^2 - \Omega_\mathrm{pf}$ and hence $M'''|_{\Sigma_\varphi=v} < 0$, except at the fixed point
where $2v^2 - \gamma(1-\Omega_\mathrm{pf})=0$ and $M$ is a constant
%
%
and hence all its derivatives are zero, exemplified by that $M'''|_{\Sigma_\varphi=v}$ with
$\Omega_\mathrm{pf} = 1 - \frac{3\gamma}{\lambda^2}$ inserted yields zero. Thus $M$ just goes through
an inflection point for the interior orbits when $\Sigma_\varphi=v$,
$\Omega_\mathrm{pf}\neq 1 - \frac{3\gamma}{\lambda^2}$.
To summarize: $M$ is monotonically decreasing in the interior state space everywhere except at the fixed point $\mathrm{S}$.
%
%
As a consequence $M$ decreases toward its positive minimum value at $\mathrm{S}$ for all interior orbits,
and thus they all end at $\mathrm{S}$. Furthermore, toward the past $M\rightarrow\infty$.
Since both $a>0$ and $a-1<0$ this implies that $\Omega_V\rightarrow 0$ or/and $\Omega_\mathrm{pf}\rightarrow 0$
toward the past for all interior orbits. It follows from the orbits on the boundaries
$\Omega_V=0$ and $\Omega_\mathrm{pf}=0$, and the local fixed point analysis, that there is one orbit entering the
state space from the fixed point $\mathrm{FL}$ and one from $\mathrm{P}$ when $\lambda < \sqrt{6}$.
The remaining interior orbits originate from the fixed point
$\mathrm{K}_-$ only when $\lambda \geq \sqrt{6}$ and from both $\mathrm{K}_-$ and $\mathrm{K}_+$
when $\lambda < \sqrt{6}$, each yielding  1-parameter set of orbits.
\end{proof}
%

%
%

This establishes the global solution structure for the present models for all $\lambda\geq 0$ and
$\gamma\in(0,2)$.\footnote{For those that are so inclined, it is possible to formalize the
above monotonicity arguments further with the \emph{Monotonicity Principle}.
The version on p. 103 in~\cite{waiell97} is as follows:
Let $\phi_t$ be a flow on $\mathbb{R}$ with $S$ an invariant set.
Let $M: S \rightarrow \mathbb{R}$ be a $C^1$ function whose range is the interval $(a,b)$, where
$a\in \mathbb{R}\cup\{-\infty\}$, $b\in \mathbb{R}\cup\{+\infty\}$ and $a<b$. If $M$ is decreasing on orbits in $S$, then
for all ${\bf x}\in S$, the $\omega$- and $\alpha$-limit sets of orbits in $S$ are given by
$\omega({\bf x}) \subseteq \{{\bf s}\in \bar{S} \backslash S| \lim_{{\bf y}\rightarrow {\bf s}} M({\bf y}) \neq b \}$ and
$\alpha({\bf x}) \subseteq \{{\bf s}\in \bar{S}\backslash S| \lim_{{\bf y}\rightarrow{\bf s}}M({\bf y}) \neq a \}$, respectively.
A more advanced version of the monotonicity principle can be found in Appendix D in~\cite{fjaetal07}.}
The solution structure is depicted in Figure~\ref{FigET_ESFB_alpha7_3_3}.
\begin{figure}[ht!]
	\begin{center}
		\subfigure[$\gamma \geq \lambda^2/3$]{\label{fig:globalI}
			\includegraphics[width=0.3\textwidth]{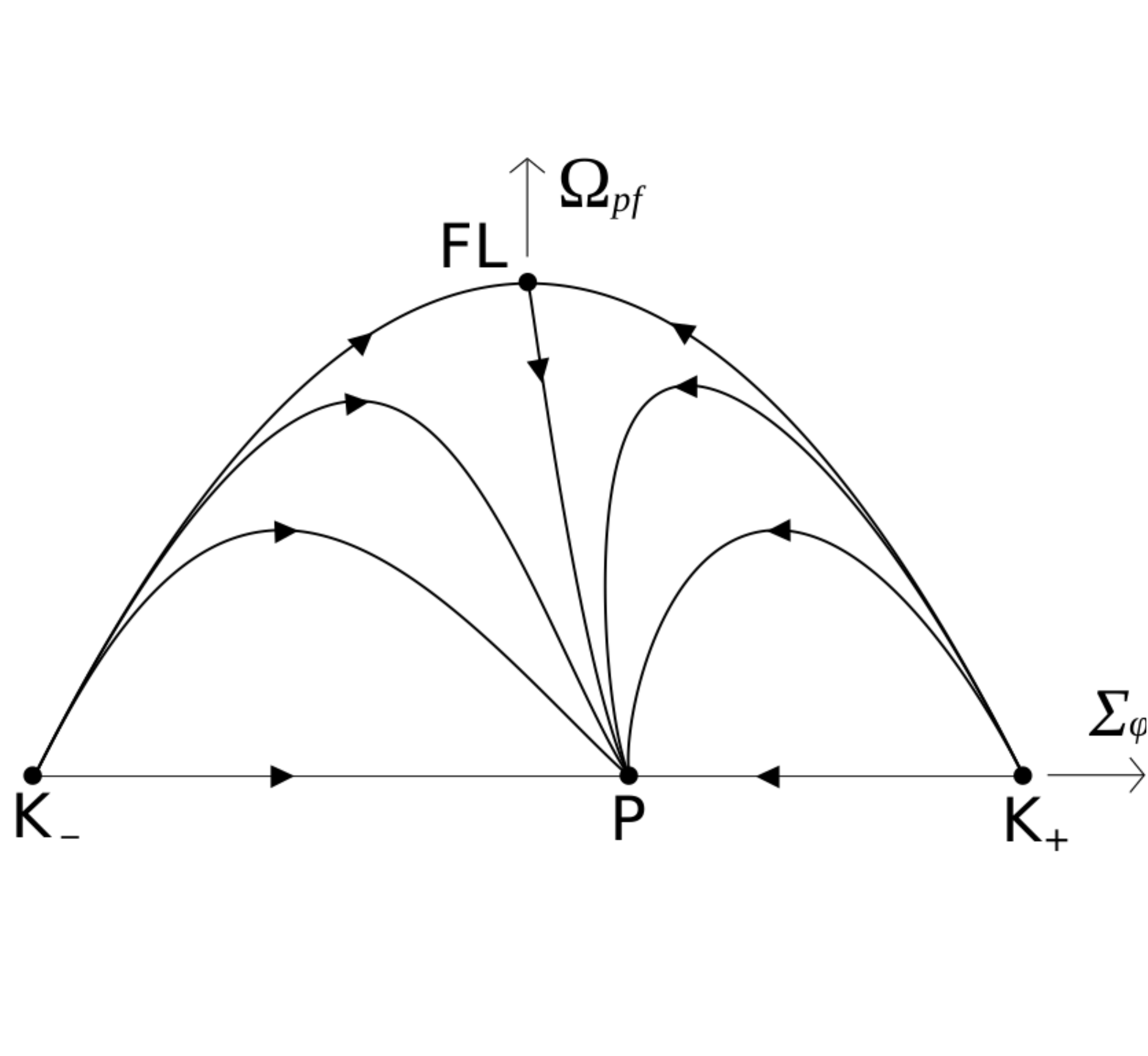}}
				\hspace{0.2cm}
		\subfigure[$\gamma<\lambda^2/3$, $\lambda<\sqrt{6}$]{\label{fig:globalII}
			\includegraphics[width=0.3\textwidth]{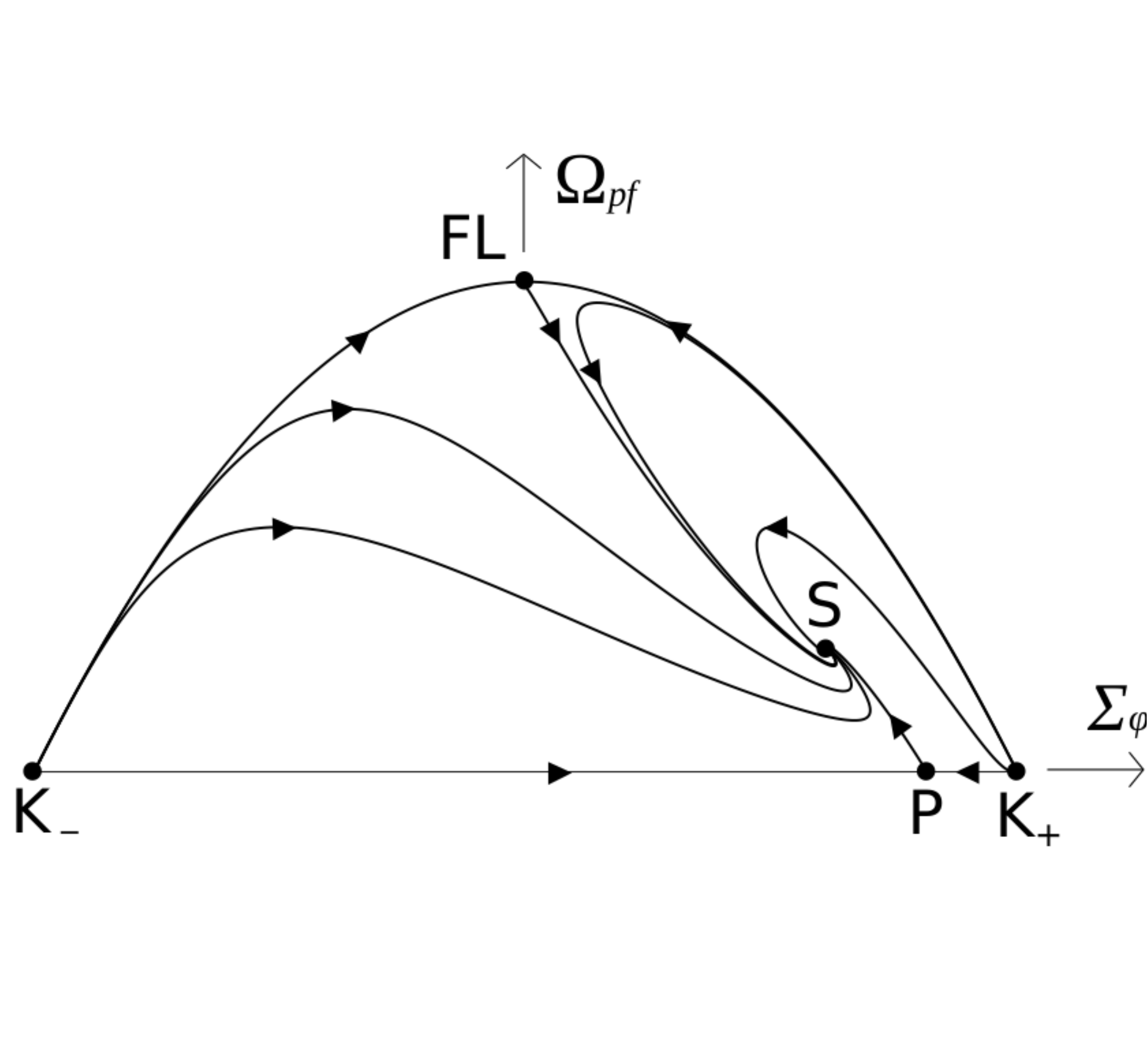}}
		\hspace{0.2cm}
				\subfigure[$\lambda \geq \sqrt{6}$]{\label{fig:globalIII}
				\includegraphics[width=0.30\textwidth]{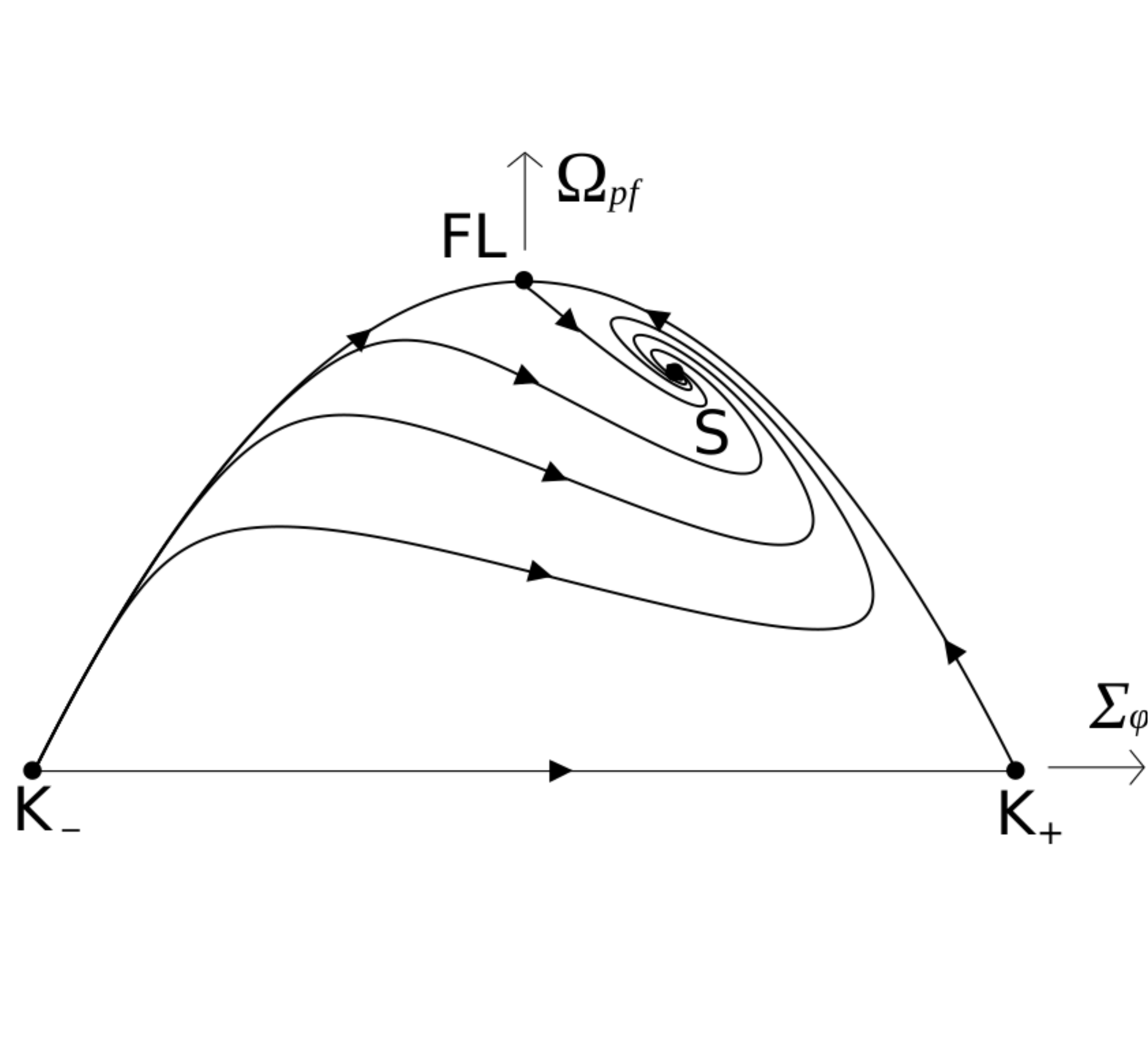}}
		\vspace{-0.5cm}
	\end{center}
	\caption{Orbit structure for a perfect fluid with linear equation of state $p_\mathrm{pf} = (\gamma-1)\rho_\mathrm{pf}$
     and an exponential/constant scalar field potential, $V = V_0e^{-\lambda\varphi}$ for the three bifurcation regions I, II, III
     and their mutual boundaries at $\gamma=\lambda^2/3$, where $\mathrm{P}$ and $\mathrm{S}$ coincide,
     and $\lambda=\sqrt{6}$, where $\mathrm{P}$ coincide with $\mathrm{K}_+$.
     }
	\label{FigET_ESFB_alpha7_3_3}
\end{figure}
However, we here offer an alternative proof using a Dulac function instead of monotonic functions.

\subsection{The Dulac function}

In contrast to the previous direct proofs, the present one relies on the Poincar\'e-Bendixson theorem.

\begin{theorem}\label{TheoremIII}
{\bf Global interior dynamics when $\gamma\in(0,2)$.}
\begin{enumerate}
The global dynamics of the dynamical system~\eqref{exppf} can be inferred from
the local stability analysis of the fixed points.
\end{enumerate}
\end{theorem}
%
%
\begin{proof}
The proof relies on the existence of a Dulac function and the simple boundary structure.
Let $(\Sigma_\varphi^\prime,\Omega_\mathrm{pf}^\prime) = (f_\Sigma,f_\Omega) = \vec{f}$, where
$f_\Sigma$ and $f_\Omega$ are the right hand sides of~\eqref{SigmaExpBound}
and~\eqref{OmegaExpBound}, respectively. In contrast to monotonic functions we do not have
a systematic method for obtaining Dulac functions. However, since they are less restrictive
for the dynamics it seems natural to make the ansatz $\Omega^{-a}_V \Omega_\mathrm{pf}^{-b}$.
This gives
\begin{equation}
\nabla\cdot(\Omega^{-a}_V \Omega_\mathrm{pf}^{-b}\vec{f}) =
\Omega^{-a}_V \Omega_\mathrm{pf}^{-b} \left[\frac32(5-2a-2b)(\gamma\Omega_\mathrm{pf}+2\Sigma_\varphi^2) + \sqrt{6}\lambda(a-1)\Sigma_\varphi -3(1+\gamma-b\gamma) \right],
\end{equation}
where $\nabla = (\partial_{\Sigma_\varphi},\partial_{\Omega_\mathrm{pf}})$.
Choosing $a=1$ to eliminate the linear $\Sigma_\varphi$ term and then $b=3/2$ to
eliminate the $\Omega_\mathrm{pf}$ and $\Sigma_\varphi^2$ terms  yield
\begin{equation}
\nabla\cdot(\Omega^{-1}_V \Omega_\mathrm{pf}^{-3/2}\vec{f}) = -\tfrac32(2-\gamma) \Omega^{-1}_V \Omega_\mathrm{pf}^{-3/2} < 0, \qquad
\gamma\in(0,2),
\end{equation}
i.e., the divergence of $\Omega^{-1}_V \Omega_\mathrm{pf}^{-3/2}\vec{f}$ is strictly
negative in the interior state space when $\gamma \in(0,2)$.
The function $\Omega^{-1}_V \Omega_\mathrm{pf}^{-3/2}$
is thereby a Dulac function, which excludes the possibility of interior periodic orbits,
see e.g. p. 265 in~\cite{per00}. In combination with that there are no heteroclinic
cycles on the boundary, it follows from the Poincar\'e-Bendixson theorem on $\mathbb{R}^2$ (see e.g.~\cite{per00} section 3.7) that the
only possible $\alpha$- and $\omega$-limit sets of the orbits
(i.e., their future and past asymptotics)
of the system~\eqref{exppf} are the fixed points, i.e.,
the local analysis of the fixed points completely describes the
future and past asymptotics of the dynamical system~\eqref{exppf}.
\end{proof}
%


In conclusion: In contrast to earlier work, we have performed a complete 
local and global analysis on FLRW models with a minimally coupled scalar field 
with an exponential potential and a perfect fluid with a linear equation of state. 
This has been accomplished by complementing previous linear analysis of hyperbolic
fixed points with a center manifold analysis of non-hyperbolic fixed points, 
associated with bifurcations, and by constructing monotonic functions and a Dulac 
function that subsequently were used to give a global description
of the dynamics of these models. The present analysis and methods serve as a 
starting point for investigations of models with more general asymptotically 
exponential potentials, and also for inhomogeneous perturbations of
such models, topics that we will come back to in future papers.

\subsection*{Acknowledgments}
It is a pleasure to thank John Wainwright for many useful discussions and for suggesting that we write
the present paper. A. A. is supported by FCT/Portugal through CAMGSD, IST-ID, projects UIDB/04459/2020 and UIDP/04459/2020.

\bibliographystyle{unsrt}
\bibliography{../Bibtex/cos_pert_papers}

\end{document}